\documentclass[prb,preprint,letterpaper,amsmath,amssymb,floatfix]{revtex4}

\usepackage{graphicx}
\usepackage{dcolumn}
\usepackage{bm}
\usepackage{eucal}

\newcommand{\smat}[1]{\mathcal{S}\left(#1\right)}
\newcommand{\smateq}[1]{\mathcal{S}_N\left(#1\right)}
\newcommand{\thetadiff}{\left(\theta_{i}-\theta_{j}\right)}
\newcommand{\thetadiffeq}{\left(\frac{2\pi}{N}\left(i-j\right)\right)}

\newcommand{\thetadiffeqf}[1]{2\pi\left(#1\right)/N}
\newcommand{\lindex}{l=1\left(l\neq k\right)}

\newcommand{\jindex}{j=1\left(j\neq k\right)}

\newtheorem{theorem}{Theorem}
\newtheorem{lemma}{Lemma}
\numberwithin{lemma}{theorem}
\newtheorem{corollary}{Corollary}
\numberwithin{corollary}{lemma}

\newenvironment{proof}[1][Proof.]{\begin{trivlist}
\item[\hskip \labelsep {\bfseries #1}]}{\end{trivlist}}

\newcommand{\qed}{\nobreak \ifvmode \relax \else
      \ifdim\lastskip<1.5em \hskip-\lastskip
      \hskip1.5em plus0em minus0.5em \fi \nobreak
      \vrule height0.4em width0.5em depth0.25em\fi}

\begin{document}

\title{Angular Normal Modes of a Circular Coulomb Cluster}

\author{L.W. Lupinski}
\author{M.J. Madsen}

\affiliation{Department of Physics, Wabash College, Crawfordsville, IN  47933}

\date{\today}

\begin{abstract}
We investigate the angular normal modes for small oscillations about an equilibrium of a single-component coulomb cluster confined by a radially symmetric external potential to a circle.  The dynamical matrix for this system is a Laplacian symmetrically circulant matrix and this result leads to an analytic solution for the eigenfrequencies of the angular normal modes.  We also show the limiting dependence of the largest eigenfrequency for large numbers of particles.
\end{abstract}

\maketitle

\section{Introduction}

We present in this paper an analytic solution for finding the angular normal mode frequencies of small oscillations about an equilibrium position for a system of interacting charged particles confined by a radially symmetric external potential.  Experimental applications of this analysis might include the following types of systems: laser cooled ions in a Penning trap\cite{gilbert}, ions in a cylindrically symmetric mass spectrometer\cite{schatz,lammert,austin}, clusters of electrons on the surface of liquid helium\cite{leiderer}, and atoms confined in a chemical ring \cite{wilson}.  We limit our investigation to two-dimensions, similar to previous work with these types of Coulomb clusters \cite{schweigert}, though three-dimensional systems have also been investigated using molecular dynamics simulations \cite{rahman}.  We also limit our investigation to a one-component plasma where particles have identical mass and charge.  Other situations where the charges are not uniform have been investigated elsewhere.\cite{barkby}

Our work can be stated in the following theorem:
\begin{theorem}
For $N$ charged particles, confined by a radially symmetric potential such that their equilibrium positions lie on a circle, the eigenfrequencies of the normal modes of oscillation for small perturbations about the equilibrium positions can be written as
$$
\omega_h^2=\displaystyle\sum_{r=2}^N\left[\smateq{1,r}\left(1-\cos\left(\frac{2\pi(h-1)(r-1)}{N}\right)\right)\right] ,
$$
where $\smateq{i,j}$  is 
$$
\smateq{i,j}=\frac{3+\cos\thetadiffeq}{4\sqrt{2-2\cos\thetadiffeq}}\frac{1}{\sin^2\left(\frac{\pi}{N}(i-j)\right)}.
$$
\label{theorem1}
\end{theorem}

In Section II we present the physical characteristics of the problem including our small-perturbation approximation in order to develop the conditions of Theorem~\ref{theorem1}.  We develop the mass-weighted dynamical (Hessian) matrix based on the inter-particle Coulomb interaction using techniques which are typically used in normal mode analysis in physical chemistry\cite{wilson}.  We utilize concepts from graph theory to construct an edge-weighted Laplacian matrix whose weights correspond to the first-order spring constant for small oscillations about the equilibrium position.  In Section III we apply the symmetry constraints on the equilibrium conditions and develop the properties of the edge weights $\smateq{i,j}$ for this physical situation.  We use these properties to prove that our dynamical matrix is a Laplacian matrix that is also symmetrically circulant.  In Section IV we build on results from Guti\'{e}rrez-Guti\'{e}rrez\cite{gutierrez} to find the eigenvalues for a weighted Laplacian symmetrically circulant matrix.  Finally, in Section V, we comment on the properties of the eigenvalues, the largest and smallest frequencies, and the limiting behavior as $N$ becomes large.
Since we limit the scope of this paper to the situation where the particle filling is low enough that there is only a single ring of particles with radius $R$ in the equilibrium state, when we talk about larger numbers of particles, we assume that the trapping potential has been appropriately scaled to maintain this condition.

\section{Physical Model}
This paper will be focused on a system of identically charged particles confined to a two dimensional circle of radius $R$, as shown in Fig.~\ref{fig1}, by an external potential.  
\begin{figure}[htbp]  
\includegraphics[width=8cm,keepaspectratio]{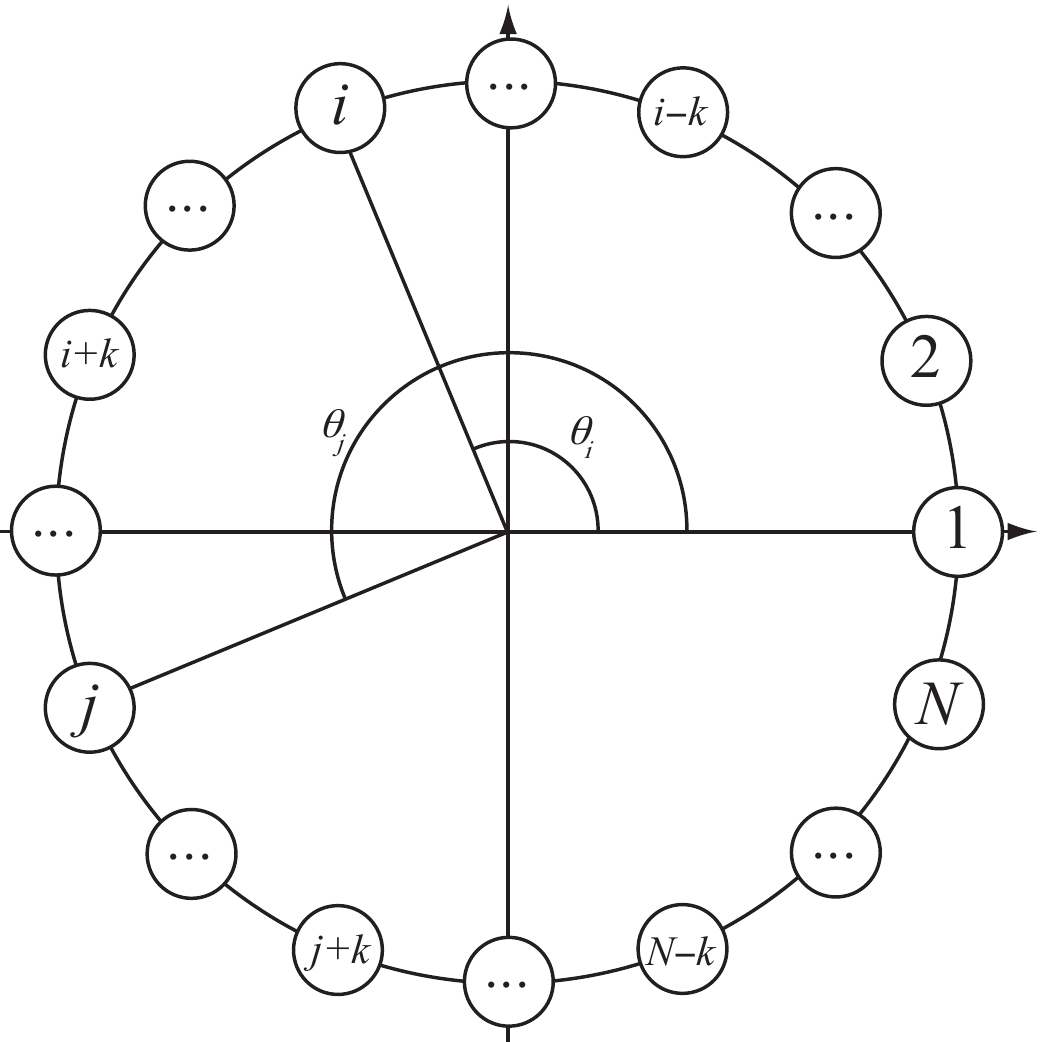}  
\caption{The equilibrium positions of $N$ particles confined by an external potential to a circle of radius $R$.  We identify each particle with an index $i$ and define the angle $\theta_i$ from the equilibrium position from the first particle.}
\label{fig1}
\end{figure}
In order to find the normal mode frequencies for small oscillations about this equilibrium position, we construct the dynamical (Hessian) matrix $\mathbf{H}$ based on the interaction Hamiltonian $H$ of this system.\cite{wilson}  The Hamiltonian has two pieces: the Coulombic repulsion between pairs of particles and a radially symmetric external potential.
\begin{equation}
H=\displaystyle\sum_{j=1}^N \displaystyle\sum_{i>j}^N  \frac{q^2}{4\pi\epsilon_0|\vec{r}_i-\vec{r}_j|} + \displaystyle\sum_{i=1}^N V(\vec{r}_i).
\end{equation}
It is natural to use polar coordinates, $\vec{r_i}=(x_i,y_i)=r_i(\cos\theta_i,\sin\theta_i)$ because of the radially symmetric potential.  The Hamiltonian in polar coordinates is thus
\begin{equation}
H=\frac{q^2}{4\pi\epsilon_0}\displaystyle\sum_{j=1}^N \displaystyle\sum_{i>j}^N  \frac{1}{\sqrt{r_i^2+r_j^2-2r_ir_j\cos\thetadiff}} + \displaystyle\sum_{i=1}^N V(r_i).
\end{equation}
We now assume that the radially symmetric external potential confines the particles to a circle with the minimum at some radius $R$.  Therefore $r_i=r_j=R$, and we obtain,
\begin{equation}
H=\frac{q^2}{4\pi\epsilon_0}\frac{1}{2R}\displaystyle\sum_{j=1}^N \displaystyle\sum_{i>j}^N 
\left(\frac{1-\cos\thetadiff}{2}\right)^{-\frac{1}{2}} +\displaystyle\sum_{i=1}^N V(R).
\end{equation}
Using
\begin{equation} 
\displaystyle\sum_{i=1}^N V(R)=V(R)+V(R)+...+V(R)=N\left(V(R)\right),
\end{equation}
the interaction Hamiltonian is thus
\begin{equation}
H=\frac{q^2}{4\pi\epsilon_0}\frac1{2R}\displaystyle\sum_{j=1}^N \displaystyle\sum_{i>j}^N   \left(\frac{1-\cos\thetadiff}{2}\right)^{-\frac12} + N\left(V(R)\right).
\end{equation}
We can now construct the mass-weighted dynamical matrix consisting of the second partial derivatives of the Hamiltonian, evaluated at the equilibrium positions of the particles.  This model is only valid for small perturbations about the equilibrium position and is only a first-order approximation.  The mass weight makes the eigenvalues of the dynamical matrix the normal mode frequencies.  Since we are working in polar coordinates, we rewrite the partial derivatives $\partial x_i\rightarrow R\partial\theta_i$, and the elements $H_{kl}$ of the dynamical matrix are
\begin{eqnarray}
H_{kl} & = & \left. \frac{1}{MR^2}\frac{\partial^2H}{\partial\theta_k\partial\theta_l}\right|_\textrm{equil}\nonumber  \\
{} & = & \frac{q^2}{4\pi\epsilon_0(2R^3M)}\frac{\partial^2}{\partial{\theta_k}\partial{\theta_l}}\left[\displaystyle\sum_{j=1}^N \displaystyle\sum_{i>j}^N   \left(\frac{1-\cos\thetadiff}{2}\right)^{-\frac12} +N V(R)\right]_\textrm{equil}
\end{eqnarray}
where $k,l$ are integers and run from $1,...,N$.  We scale the dynamical matrix  by the natural frequency for this system, $\omega_0^2=q^2/4\pi\epsilon_0(2R^3M)$ which is equivalent to scaling the Hamiltonian by the harmonic oscillator frequency of the confining potential.  Therefore all eigenfrequencies are scaled in these units as well.  We also note that the $N V(R)$ term is only dependent on $R$ and so does not contribute after the partial differentiation.  The scaled matrix elements are thus
\begin{equation}
H_{kl}=\frac{\partial^2}{\partial{\theta_k}\partial{\theta_l}}\left[\displaystyle\sum_{j=1}^N \displaystyle\sum_{i>j}^N   \left(\frac{1-\cos\thetadiff}{2}\right)^{-\frac12} \right] =\displaystyle\sum_{j=1}^N \displaystyle\sum_{i>j}^N  \frac{\partial^2}{\partial{\theta_k}\partial{\theta_l}}  \left(\frac{1-\cos\thetadiff}{2}\right)^{-\frac12}.
\label{hesspartial1}
\end{equation}
We can now simplify the partial derivatives using the Kronecker delta notation:
\begin{equation}
\frac{\partial}{\partial\theta_l}   \left(\frac{1-\cos\thetadiff}{2}\right)^{-\frac12} = \frac12 \frac{\cot\left(\frac{\theta_i-\theta_j}{2}\right)}{\sqrt{\frac{1-\cos\thetadiff}{2}}} (\delta_{lj}-\delta_{li}) .
\end{equation}
We then take the second derivative, giving 
\begin{equation}
\frac{\partial}{\partial\theta_k}\left(\frac12 \frac{\cot\left(\frac{\theta_i-\theta_j}{2}\right)}{\sqrt{\frac{1-\cos\thetadiff}{2}}} (\delta_{lj}-\delta_{li}) \right) = \frac{3+\cos\thetadiff}{4\sqrt{2-2\cos\thetadiff}}\frac{1}{\sin^2\left(\frac{\theta_i-\theta_j}{2}\right)}(\delta_{kj}-\delta_{ki})(\delta_{lj}-\delta_{li}).
\end{equation}
Substituting this result into (\ref{hesspartial1}) we obtain,
\begin{equation}
H_{kl} =  \displaystyle\sum_{j=1}^N \displaystyle\sum_{i>j}^N \left[\frac{3+\cos\thetadiff}{4\sqrt{2-2\cos\thetadiff}}\frac{1}{\sin^2\left(\frac{\theta_i-\theta_j}{2}\right)}(\delta_{lj}\delta_{kj}-\delta_{lj}\delta_{ki}-\delta_{li}\delta_{kj}+\delta_{li}\delta_{ki})\right]_\textrm{equil}
\label{kdeltaform}.
\end{equation}

In order to simplify our notation we will define the function $\mathcal{S}$ such that,
$$
\smat{i,j} = \frac{3+\cos\thetadiff}{4\sqrt{2-2\cos\thetadiff}}\frac{1}{\sin^2\left(\frac{\theta_i-\theta_j}{2}\right)}.
$$
In graph theory $\smat{i,j}$ is the weighted edge function of the inter-connected particle interaction graph.  It is important to note that $\smat{i,j}$ is an even function which leads to the following lemma:

\begin{lemma}
The edge weight function $\smat{i,j}$ has the property that $\smat{i,j} = \smat{j,i}$.
\label{symmetric}
\end{lemma}

\begin{proof}
For the cosine components of $\smat{i,j}$,
$\cos\left(\theta_i-\theta_j\right)$ is equivalent to $\cos\left(-\left(\theta_j-\theta_i\right)\right)$ and because cosine is an even function is equivalent to $\cos\left(\theta_j-\theta_i\right)$. 
By symmetry this holds for the $\sin^2$ component, because it is also even.  Therefore $\smat{i,j} = \smat{j,i}$.
$\qed$
\end{proof}
This also means that $H_{kl}=H_{lk}$, which agrees with the property that Hessian matrices are symmetric.  This result is a consequence of Newton's third law since the interactions between two particles are equal.

We will now define the edge weight function, $\smateq{i,j}$, evaluated at the equilibrium positions of the $N$ particles on the circle.  The particles are equally distributed around the circle such that the difference in the angular position between any two particles is a multiple of $2\pi/N$.  Therefore an arbitrary particle identified by the index $i$ has an equilibrium angular position of $\theta_i=2\pi i/N+\theta_0$ where $\theta_0$ is the arbitrary starting point on the circle (given by the angle of the first particle).  Thus the difference in angular position between particles $i$ and $j$ is  $\left(\theta_i-\theta_j\right)=\thetadiffeqf{i-j}$.

Now that the equilibrium geometry is established we can define, the edge weight for N particles
\begin{equation}
\smateq{i,j}=\frac{3+\cos\thetadiffeq}{4\sqrt{2-2\cos\thetadiffeq}}\frac{1}{\sin^2\left(\frac{\pi}{N}(i-j)\right)}.
\end{equation}
An illustration of the first five interconnected particle graphs, connected by edges with the weight of $\smateq{i,j}$, is shown in Fig.~\ref{fig2}.
\begin{figure}[htbp]  
\includegraphics[width=6in,keepaspectratio]{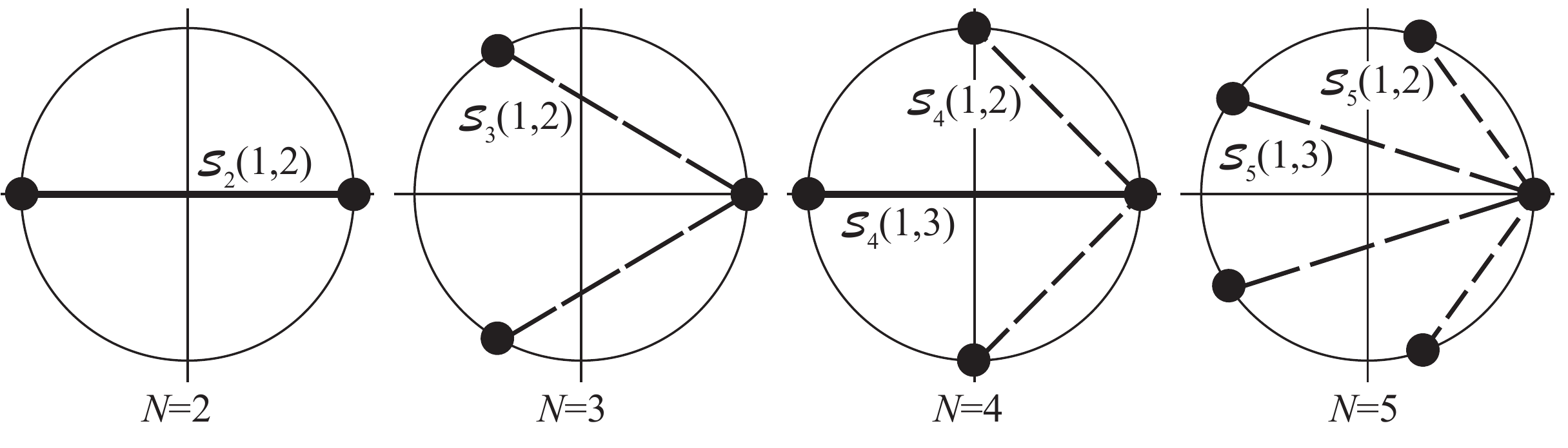}  
\caption{The first five interconnected particle graphs showing one set of the weighted edges,  $\smateq{1,j}$ where $j=1,...,N$.}
\label{fig2}
\end{figure}
Using $\smateq{i,j}$ we can further simplify (\ref{kdeltaform}).
Beginning with the $k \neq l$ case and evaluating the Kronecker deltas, we are left with only two possibilities
\begin{equation}
H_{kl}=
\begin{cases} 
-\smateq{k,l} & \mbox{if } \mbox{$k > l$} \\ 
-\smateq{l,k} & \mbox{if } \mbox{$k < l$} .
\end{cases}
\end{equation}
We can then combine this statement using Lemma \ref{symmetric}.  Therefore
the statement above can be written as
\begin{equation}
H_{kl} = -\smateq{k,l}   \mbox{for all $k \neq l$}.
\label{kneql}
\end{equation}
Now we will look at the $k = l$ case, from $\left(\ref{kdeltaform}\right)$,
\begin{equation}
H_{kk}=
\begin{cases} 
\displaystyle\sum_{j=1}^N \smateq{k,j} = \displaystyle\sum_{j=1}^{k-1} \smat{k,j}  & \mbox{if } \mbox{$k > j$} \\ 
 \displaystyle\sum_{i > k}^N \smateq{i,k} =  \displaystyle\sum_{i=k+1}^N \smat{i,k} & \mbox{if } \mbox{$k < i$} .
\end{cases}
\end{equation}
Recognizing the overlap between the two summations, and again using Lemma \ref{symmetric}, we combine them to give
\begin{equation}
H_{kk} = \displaystyle\sum_{\jindex}^N\smateq{k,j} 
\label{keql}
\end{equation}
where the $\left(j \neq k\right)$ notation means that the term $j=k$ is skipped when performing the sum.  Now (\ref{kneql}) and (\ref{keql}) describe each entry of the matrix.

\section{Properties of $\smateq{i,j}$}
Before we can prove that the weighted Hessian matrix generated by functions of $\smateq{i,j}$ is a Laplacian symmetric circulant matrix, we will first prove several properties of $\smateq{i,j}$.  We already know from Lemma~\ref{symmetric}, that $\smateq{i,j}=\smateq{j,i}$ and that $\smateq{i,j} $ is an even function.  The new properties are contained in the following lemmas and their corollaries,  with a graphical representation of the properties displayed in these proofs shown in Fig.~\ref{fig3}.

\begin{figure}[htbp]  
\includegraphics[width=8cm,keepaspectratio]{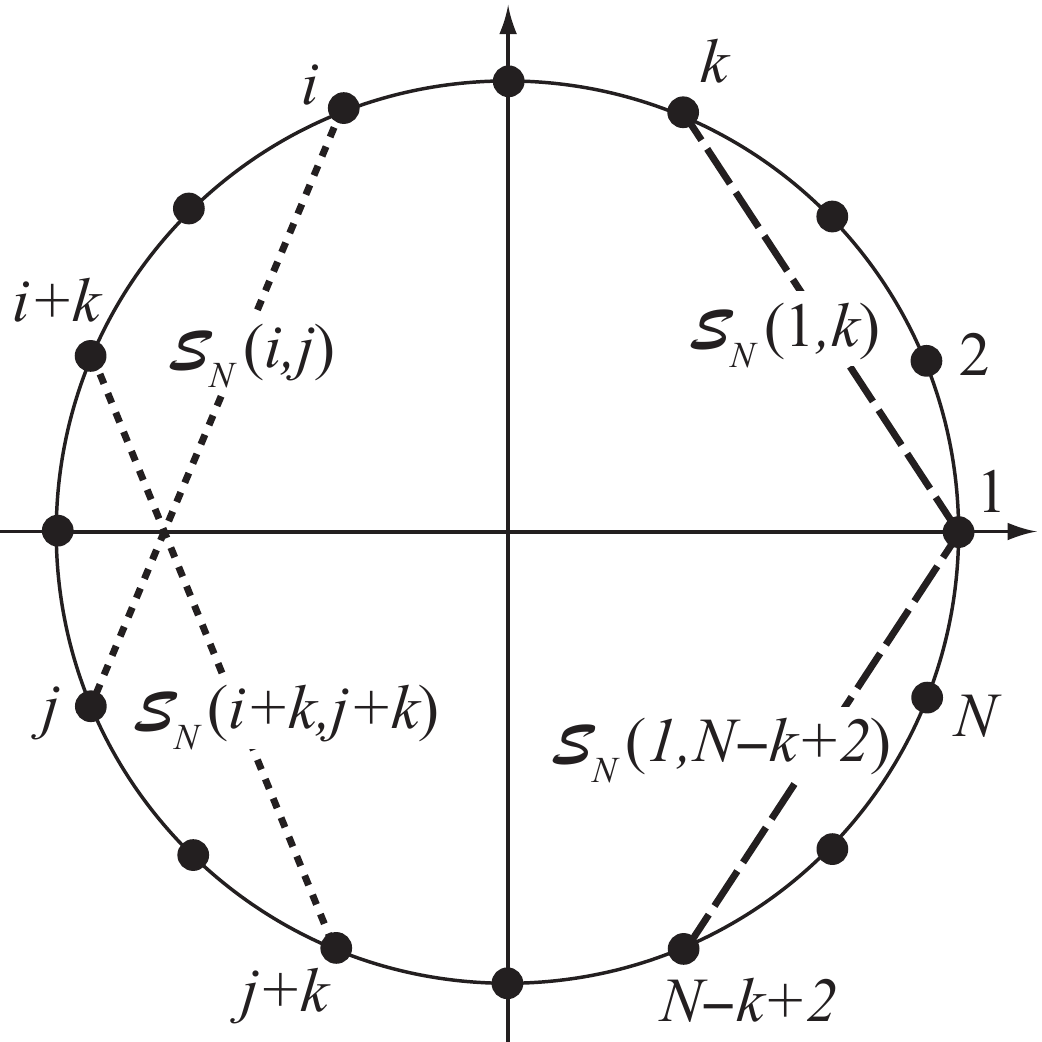}  
\caption{This diagram explicitly demonstrates Lemmas \ref{Toeplitz} and \ref{Lemma2_1}.  The setup from which the other Lemmas can be graphically interpreted is also shown.}  
\label{fig3}
\end{figure}

\begin{lemma}
\label{Toeplitz}
The edge weight function  $\smateq{i,j}$ has the property $\smateq{i,j} = \smateq{i+k,j+k}$.
\end{lemma}
\begin{proof}
The angle of the trig functions corresponding to $\smateq{i,j}$ is $\thetadiffeqf{i-j}$.  Adding and subtracting a constant $k$ from the argument gives us  $\thetadiffeqf{\left(i+k\right)-\left(j+k\right)}$, which is the angle corresponding to $\smateq{i+k,j+k}$.  Therefore $\smateq{i,j} = \smateq{i+k,j+k}$. $\qed$
\end{proof}

This result is a consequence that the interaction between any two particles separated by the same distance  is the same (as shown in Fig.~\ref{fig3}) and corresponds to the discrete rotational symmetry of the physical system.

\begin{corollary}
\label{for_backward}
The edge weight function $\smateq{i,j}$ has the property $\smateq{i,i+k}=\smateq{i,i-k}$.
\end{corollary}
\begin{proof}
By Lemma \ref{symmetric} $\smateq{i,i+k}$ is equivalent to $\smateq{i+k,i}$.  Using Lemma \ref{Toeplitz}, we can add an $i-k$ to both sides, which yields $\smateq{i,i-k}$.  Therefore $\smateq{i,i+k}=\smateq{i,i-k}$. $\qed$
\end{proof}

\begin{lemma}
\label{piper2}
The edge weight function $\smateq{i,j}$ has the property $\smateq{0,j}=\smateq{N,j}$.
\end{lemma}
\begin{proof}
The angle of the trig functions corresponding to $\smateq{0,j}$ is $-\thetadiffeqf{j}$.  Also the angle corresponding to $\smateq{N,j}$ is $-\thetadiffeqf{N-j}=2\pi-\left(2\pi{j}/N\right)$.  Therefore because the trigonometric components of $\smateq{i,j}$ are $2\pi$ periodic, $\smateq{0,j}=\smateq{N,j}$.  $\qed$
\end{proof}
There is no $0^\mathrm{th}$ particle on the circle; the physical interpretation of the Lemma is that one complete journey around the circle is achieved by a rotation of $N$ particles.  For instance, starting at particle 1 on Fig.~\ref{fig3} and moving around N particles forward, brings one back to particle 1.  
\begin{corollary}
The edge weight function $\smateq{i,j}$ has the property $\smateq{i,k}=\smateq{i,k\pm{N}}$.
\label{nzero}
\end{corollary}
\begin{proof}
The function $\smateq{i,k-N}$ by Corollary \ref{for_backward} is equivalent to $\smateq{i,k+N}$.  Adding $-k$ to both sides, by Lemma \ref{Toeplitz} we obtain $\smateq{i-k,N}$.  Which, by Lemma \ref{piper2}, is equal to  $\smateq{i-k,0}$.  Adding $k$ back to both sides yields $\smateq{i,k}$. $\qed$
\end{proof}

\section{Laplacian Symmetric Circulant Matrix}
Our analytic solution requires that the weighted Hessian $\mathbf{H}$ be a symmetric circulant matrix.  We gain other insights about $\mathbf{H}$ if it is Laplacian.  In this section we prove that the matrix is a Laplacian symmetric circulant matrix.
\begin{theorem}
The matrix $\mathbf{H}$ is a Laplacian matrix, meaning that the sum of each row and column is zero.
\end{theorem}
\begin{proof}
Because $\mathbf{H}$ is symmetric it suffices to show that the sum of each row is zero.  The sum of each row can be written as 
$$
\displaystyle\sum_{l=1}^N  H_{kl}.
$$
Taking the $l=k$ term out of the sum we obtain 
$$
H_{kk}+\displaystyle\sum_{\lindex}^N  H_{kl}.$$
We can now substitute in the results from Eqns. (\ref{kneql}) and (\ref{keql}); the sum of the row becomes 
$$
\left(\displaystyle\sum_{\lindex}^N\smateq{k,l}\right) + \displaystyle\sum_{\lindex}^N (-\smateq{k,l})=\displaystyle\sum_{\lindex}^N\smateq{k,l} - \displaystyle\sum_{\lindex}^N \smateq{k,l}=0.
$$
Therefore the sum of each row is zero.  Thus the matrix $\mathbf{H}$ is a Laplacian matrix. \qed
\end{proof}
The properties of the eigenvalues of a Laplacian matrix are such that all of the eigenvalues are real and positive and at least one eigenvalue will be zero.  The zero eigenfrequency for the normal mode oscillations corresponds to the common rotational mode of motion of the particles moving in the same direction around the circle.\cite{chung}

\begin{theorem}
The matrix $\bf{H}$ is symmetrically circulant, given by the form\cite{gutierrez}
$$
\mathrm{circ}_N(c_0,c_1, \cdots, c_{N-1})=
\begin{bmatrix}
c_0 & c_{N-1} & c_{N-2} & \cdots & c_1 \\
c_{1} & c_0 & c_{N-1} & \ddots & \vdots \\
c_2 & c_1 & c_0 & \ddots & c_{N-2} \\
\vdots & \ddots & \ddots & \ddots & c_{N-1}\\
c_{N-1} & ... & c_2 & c_1 & c_0
\end{bmatrix}.
$$
A symmetric circulant matrix matrix, $B_n$, has the form that
$$
B_n = 
\begin{cases} 
\mathrm{circ}_N(b_0,b_1,\cdots,b_{\frac{n-1}{2}},b_{\frac{n-1}{2}},\cdots,b_1) & \mbox{if } n\mbox{ is odd} \\ 
\mathrm{circ}_N(b_0,b_1,\cdots,b_{\frac{n}{2}-1},b_{\frac{n}{2}},b_{\frac{n}{2}-1},\cdots,b_1) & \mbox{if } n\mbox{ is even.} \end{cases}
$$
\label{thm3}
\end{theorem}

To prove that matrix, $\mathbf{H}$ is a symmetric circulant matrix we need only prove that all of its diagonals are the same and the rows and columns follow the pattern for both odd and even cases.  Physically we know that because all of the particles are identical the diagonals of the matrix should be the same.

\begin{lemma}
The matrix $\mathbf{H}$ is a Toeplitz  matrix, therefore the entries of each left-to-right diagonal is uniform.
\end{lemma}

\begin{proof}
This proof will cover the different two cases for $H_{kl}$: $k \neq l$ and $k = l$.  

For $k\neq l$ we know that from Lemma~\ref{Toeplitz} $\smateq{i,j} = \smateq{i+k,j+k}$.  Therefore, $H_{kl}$ is equal to $H_{ij}$ where $i=k+m$ and $j=l+m$ for some integer $m$.  Since a left-to-right diagonal is described by stepping both indices by the same number, the diagonal entries in $\mathbf{H}$ are the same for $k \neq l$.

For $k = l$, we will show that an arbitrary entry, $H_{kk}$ on the diagonal is equal to the next entry down the diagonal, $H_{(k+1)(k+1)}$.  From Eq.~(\ref{keql}) we know that the central diagonal is of the form 
\begin{equation}
H_{kk} = \displaystyle\sum_{\jindex}^N\smateq{k,j}= \smateq{k,1}+\smateq{k,2}+...+\smateq{k,N-1}+\smateq{k,N}.
\label{eq16}
\end{equation}
where $k\neq{j}$.  Also, the next entry down the diagonal is similarly
\begin{eqnarray}
H_{(k+1)(k+1)} &= & \displaystyle\sum_{\jindex}^N\smateq{k+1,j} \\
&=& \smateq{k+1,1}+\smateq{k+1,2}+...+\smateq{k+1,N-1}+\smateq{k+1,N}\nonumber.
\end{eqnarray}
Using Lemma \ref{Toeplitz} we can add $-1$ to $j$ and $k$, giving 
\begin{eqnarray}
H_{(k+1)(k+1)} & =  & \displaystyle\sum_{j=1(j\neq k+1)}^N\smateq{k,j-1} \nonumber\\ 
&=& \smateq{k,0}+\smateq{k,1}+...+\smateq{k,N-2}+\smateq{k,N-1}.
\label{eq18}
\end{eqnarray}
All of the terms between $H_{kk}$ and $H_{(k+1)(k+1)}$ from (\ref{eq16}) and (\ref{eq18}) overlap except $\smateq{k,N}$ from  (\ref{eq16}) and  $\smateq{k,0}$ from (\ref{eq18}).  However from Lemma \ref{piper2},  $ \smateq{k,0}=\smateq{k,N}$. Therefore $H_{kk}=H_{(k+1)(k+1)}$, and all of the central diagonals are the same. \qed
\end{proof}

In order to show that the matrix is symmetrically circulant, we also need the following lemma:
\begin{lemma}
The first row of the matrix $\mathbf{H}$ has the symmetry such that $\smateq{1,k} = \smateq{1,N-k+2}$.
\label{Lemma2_1}
\end{lemma}
\begin{proof}
By Corollary \ref{nzero}, $\smateq{1,N-k+2}$ is equivalent to $\smateq{1,-k+2}$.  Using Lemma \ref{Toeplitz} we can add $k-1$ to both sides, which yields $\smateq{k,1}$, which by Lemma \ref{symmetric} is equal to $\smateq{1,k}$.
Therefore $\smateq{1,k} = \smateq{1,N-k+2}$. \qed
\end{proof}
This property is demonstrated as the connection between the first particle and the $k^\mathrm{th}$ particle going forwards and backwards around the circle in Fig.~\ref{fig3}.

With this proof, we have shown that the dynamical matrix $\bf{H}$ is symmetrically circulant of the form given in Theorem~\ref{thm3}. $\qed$

We now use the results by Guti\'{e}rrez-Guti\'{e}rrez and extend those results to find the eigenvalues for a Laplacian symmetrically circulant matrix. The eigenvalues $\omega^2_h$ ($h=1,\ldots,N$) of a symmetrically circulant matrix are,\cite{gutierrez}
\begin{equation}
\omega^2_h=
\begin{cases} 
 H_{11}+2\displaystyle\sum_{r=2}^{\frac{N+1}{2}}  H_{r1}\cos\left(\frac{2\pi(h-1)(r-1)}{N}\right)  & \mbox{if } N\mbox{ is odd} \\ 
H_{11}+2\displaystyle\sum_{r=2}^{\frac{N}{2}}  H_{r1}\cos\left(\frac{2\pi(h-1)(r-1)}{N}\right)  + H_{\left(\frac{N}{2}+1\right)1}\cos(\pi(h-1)) & \mbox{if } N\mbox{ is even}.
\end{cases}
\label{eigenval1}
\end{equation}
These two cases can be simplified to a single case by recognizing that the sum is doubled.  Since the matrix has the form shown in Theorem~\ref{thm3}, the sums in Eq.~(\ref{eigenval1}) are equivalent to a single sum over the entire row.  There is no longer a distinction between the even and odd cases, giving a single sum: 
\begin{equation}
\omega^2_h=H_{11}+\displaystyle\sum_{r=2}^{N}  H_{r1}\cos\left(\frac{2\pi(h-1)(r-1)}{N}\right).
\label{eigenval2}
\end{equation}
We use our results for the explicit form of the dynamical matrix, Eqs.~(\ref{kneql}) and (\ref{keql}), and that $\mathbf{H}$ is symmetric,  and find that
\begin{equation}
\omega^2_h=\displaystyle\sum_{r=2}^N\smateq{1,r}-\displaystyle\sum_{r=2}^{N}  \smateq{1,r}\cos\left(\frac{2\pi(h-1)(r-1)}{N}\right).
\end{equation}
We then combine the two sums, since they sum over the same index and range.  We also note that the eigenfrequencies have the same structure as the first row of the matrix.  There are thus only $\lceil(N+1)/2\rceil$ unique eigenfrequencies, where the notation $\lceil \rceil$ denotes the ceiling function, meaning that the number is rounded up to the nearest integer.  The unique eigenfrequencies are thus
\begin{equation}
\omega_h^2=\displaystyle\sum_{r=2}^{N}\left[\smateq{1,r}\left(1-\cos\left(\frac{2\pi(h-1)(r-1)}{N}\right)\right)\right]\; \mathrm{for}\;h\leq \lceil(N+1)/2\rceil.
\label{eigenvaluesfinal}
\end{equation}

\section{Eigenfrequencies}

Finally, we evaluate the normal mode eigenfrequencies from Eq.~(\ref{eigenvaluesfinal}) for large numbers of particles.  Since there are approximately $N/2$ eigenvalues for $N$ particles, we evaluated the eigenfrequencies on a logarithmic scale, shown for the first 10,000 particles in Fig.~\ref{modfig}.  We have omitted the zero eigenfrequencies, as noted above, since that corresponds to the common rotation of all the particles around the circle.  There is a clear frequency gap between the second and third frequency which means that, experimentally, these will always be uniquely addressable.  However, spacing between the higher mode frequencies becomes smaller as the number of particles increases.

The highest frequency mode is also of interest and we show a linear fit to this mode in the figure.  The limiting behavior of this frequency is dominated by the denominator of the expression in (\ref{eigenvaluesfinal}).  This denominator can be rewritten as the absolute value of a sine cubed, giving rise to an $N^3$ dependence of the frequency.  The solid line is a fit to $\omega^2 = \xi N^3$ with a fit parameter of $\xi = 0.067838$.  Although the analytic evaluation of the highest mode frequency ($h\rightarrow N/2$ in Eq.~(\ref{eigenvaluesfinal})) is not trivial, is is possible to show that the fit parameter approaches 
\begin{equation}
\frac{14}{8\pi^3}\zeta(3)
\end{equation}
as $N\rightarrow\infty$ where $\zeta(3)$ is the Riemann Zeta function.

\begin{figure}[htbp]  
\includegraphics[width=6in,keepaspectratio]{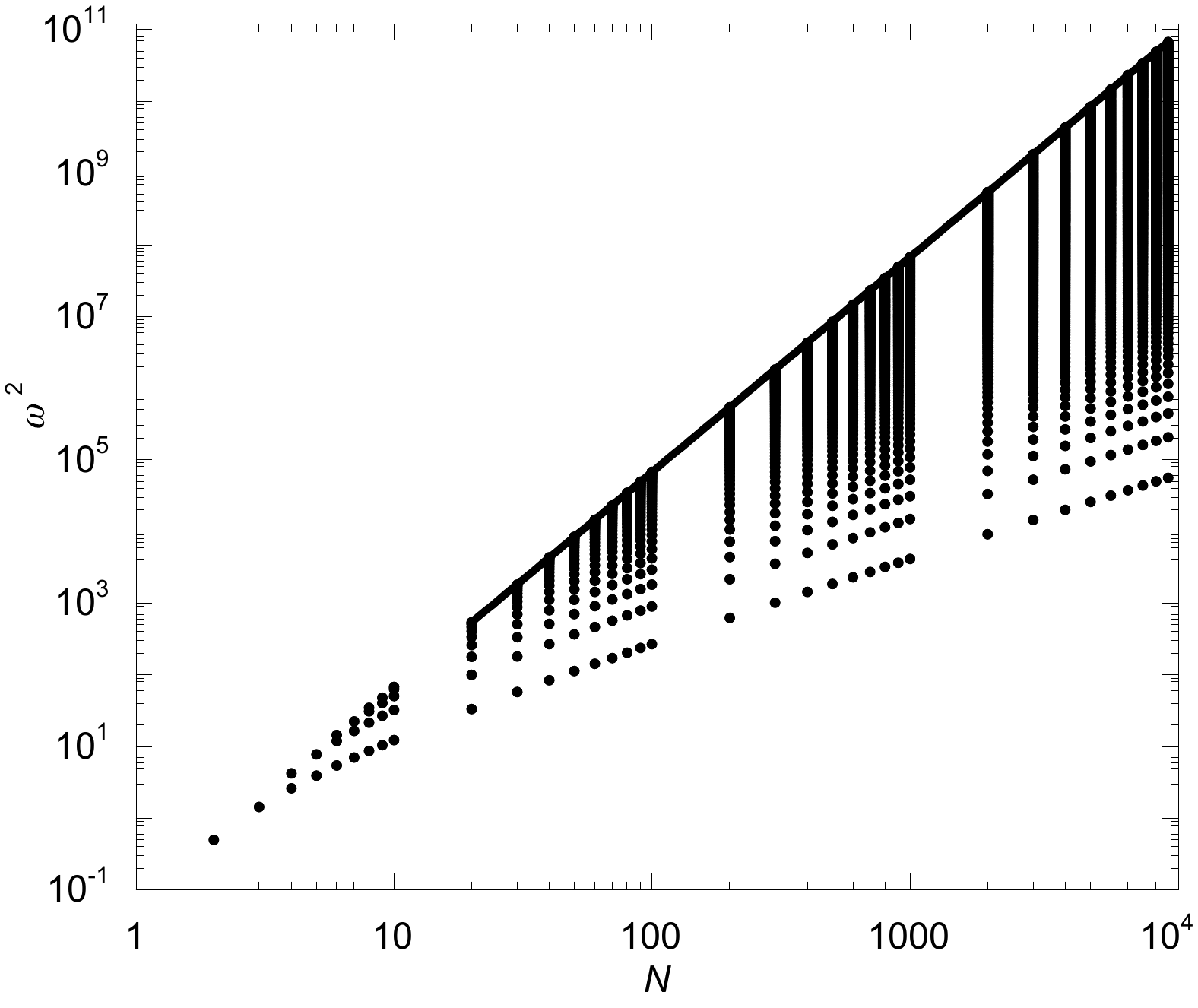}  
\caption{Eigenfrequencies for the first 10,000 particles, as calculated from Eq.~(\ref{eigenvaluesfinal}).  The solid line is a fit to $N^3$ for the highest frequency mode for large numbers of particles ($N \gg 1$).} 
\label{modfig}
\end{figure}

\section{Conclusion}
We have proven Theorem 1, for $N$ number of charged particles confined by a cylindrically symmetric potential to a circle.  Small perturbations about the equilibrium positions give rise to normal modes of oscillation which are, 
$$
\omega_h^2=\displaystyle\sum_{r=2}^N\left[\smateq{1,r}\left(1-\cos\left(\frac{2\pi(h-1)(r-1)}{N}\right)\right)\right] ,
$$
where $\smateq{i,j}$  is 
$$
\smateq{i,j}=\frac{3+\cos\thetadiffeq}{4\sqrt{2-2\cos\thetadiffeq}}\frac{1}{\sin^2\left(\frac{\pi}{N}(i-j)\right)}.
$$

\end{document}